\documentclass[a4paper]{article}
\RequirePackage{fullpage}

\usepackage[usenames,dvipsnames]{xcolor}
\usepackage[breaklinks,colorlinks, citecolor={BlueViolet}, linkcolor={Blue},urlcolor=Maroon]{hyperref}
\usepackage{amsmath}
\usepackage{amsfonts}
\usepackage{amsthm}
\usepackage{systeme, mathtools} \usepackage{tikz,pgfkeys}
\usetikzlibrary{quotes}
\usepackage[shortlabels]{enumitem}
\usepackage{subcaption}

\newtheorem{theorem}{Theorem}[section]
\newtheorem{lemma}[theorem]{Lemma}
\newtheorem{corollary}[theorem]{Corollary}
\newtheorem{proposition}[theorem]{Proposition}
\newtheorem*{definition}{Definition}

\tikzset{
  filled vertex/.style={circle, draw=blue, fill=black!50, inner sep=1pt},
  empty vertex/.style={circle, draw, fill=white, inner sep=1.5pt, minimum width=1.5pt}
}

\newcommand{\OPT}{\ensuremath{\operatorname{opt}}}
\newcommand{\ce}{\ensuremath{\operatorname{ce}}}
\newcommand{\btt}{\ensuremath{\operatorname{btt}}}
\newcommand{\Gt}{\ensuremath{G\vee K_{t}}} \newcommand{\Gm}{\ensuremath{G\vee K_{m}}} 

\title{Cluster Deletion is as Hard to Approximate as Vertex Cover}

\author{Yixin Cao\thanks{Department of Computing, Hong Kong Polytechnic University, Hong Kong, China.  \texttt{yixin.cao@polyu.edu.hk}. Supported in part by the National Natural Science Foundation of China (NSFC) under grant 62372394.} \and Ying Xu\footnotemark[1]}

\begin{document}

\maketitle

\begin{abstract}
  Recent breakthroughs in Cluster Editing have motivated attempts to adapt these approaches to obtain better-than-$2$ approximations for Cluster Deletion. We rule out this possibility under the Unique Games Conjecture: Cluster Deletion is NP-hard to approximate within a factor of $2-\epsilon$ for every fixed $\epsilon>0$, matching the known $2$-approximation [Veldt et al., WWW 2018].  Our approximation-preserving reduction from Vertex Cover also implies NP-hardness of approximation within $\sqrt2-\epsilon$.  We also show that better-than-$2$ approximations are possible in restricted settings.

  We close the paper with a brief discussion of the relationship between Cluster Editing and Bad Triangle Transversal. In particular, we give a $31$-vertex graph~$G$ for which the two optimal values differ, answering an open question of Adriaens and Tatti [ICML 2026].
\end{abstract}

\section{Introduction}

A \emph{cluster graph} is a graph in which every component is a clique.
Given a graph, the \emph{Cluster Deletion} problem asks for a minimum number of edges whose deletion results in a cluster graph.
Equivalently, Cluster Deletion asks for a partition of the vertex set into cliques that minimizes the number of edges crossing between different parts.
The problem was first motivated by applications in clustering gene networks~\cite{ben-dor-99-clustering-gene}.
More broadly, it arises as a natural special case of graph-based clustering frameworks in which pairwise similarities are represented by edges and the desired output is a decomposition into internally consistent groups~\cite{charikar-04-approximate-clustering, veldt-18-community-detection}.

A closely related problem is \emph{Cluster Editing}, where one is allowed to both insert and delete edges, and the objective is to minimize the total number of modifications.
Van Zuylen and Williamson~\cite{zuylen-09-deterministic-approximation} studied a common generalization of these problems: the input consists of a graph together with a set of friendly pairs and a set of hostile pairs, and the goal is to find a minimum cluster editing solution that does not delete any friendly pair or add any hostile pair.
We refer to this problem as \emph{Constrained Clustering}, also known as Constrained Correlation Clustering.
Constrained Clustering can be modeled as weighted Cluster Editing by assigning prohibitively large weights to deleting friendly pairs and to adding hostile pairs.
Cluster Deletion is the special case in which all non-edges are hostile.

Both Cluster Editing and Cluster Deletion are NP-complete and have been studied extensively from the perspective of approximation algorithms.  For Cluster Editing, early approximation algorithms were based on the canonical linear programming (LP) relaxation, culminating in a ratio of $2.06$~\cite{chawla-15-correlation-clustering}, close to the integrality gap of~$2$ for this LP~\cite{charikar-04-approximate-clustering}.
Attempts to close this small gap within the same framework were unsuccessful.
Cohen-Addad et al.~\cite{cohen-addad-22-sherali-adams} instead used a constant number of rounds of the Sherali--Adams hierarchy to break the factor-$2$ barrier, obtaining a $(1.994+\epsilon)$-approximation; this was later improved to $1.73+\epsilon$~\cite{cohen-addad-23-preclustering}.  The latest progress is due to Cao et al.~\cite{cao-24-lp,cao-25-sublinear}, who introduced a new LP relaxation and obtained an approximation ratio of $1.485$.\footnote{The conference versions claimed a slightly better ratio, but a bug was later found in the proof; see arXiv:2404.17509 and arXiv:2503.20883.}
Garc\'{i}a-Soriano and Schohn~\cite{garcia-soriano-26-1.387} recently announced a $(1.3865+\epsilon)$-approximation based on the same LP.

Most known approximation algorithms for Cluster Deletion have been derived from techniques for Cluster Editing.  Charikar et al.~\cite{charikar-04-approximate-clustering} showed that their algorithm can be adapted to give a $4$-approximation for Cluster Deletion.  Van Zuylen and Williamson~\cite{zuylen-09-deterministic-approximation} derandomized the randomized rounding scheme of Ailon et al.~\cite{ailon-08-aggregating-inconsistent-information}, obtaining a $3$-approximation, and showed that the same approach applies to Constrained Clustering. Veldt et al.~\cite{veldt-18-community-detection} used a refined rounding analysis to improve the approximation ratio for Cluster Deletion to $2$.

Given the recent breakthroughs for Cluster Editing, it is natural to ask whether the ideas of~\cite{cohen-addad-22-sherali-adams,cohen-addad-23-preclustering,cao-24-lp, cao-25-sublinear} can be adapted to obtain a better-than-$2$ approximation for Cluster Deletion.  Kalavas et al.~\cite{kalavas-25-constrained-correlation-clustering} reported partial progress in this direction.  However, we show that the factor-$2$ barrier for Cluster Deletion is of a very different nature from that for Cluster Editing.

On the negative side, Shamir et al.~\cite{shamir-04-graph-modification-problem} proved that, assuming $\mathrm{P}\ne\mathrm{NP}$, there is a constant $c>1$ such that no polynomial-time algorithm can approximate Cluster Deletion within a factor smaller than~$c$.  Dessmark et al.~\cite{dessmark-07-cluster-edge-deletion} showed that $c\ge 881/880$; see also \cite{adriaens-26-bad-triangle-transversals}.  We improve this lower bound to~$\sqrt{2}$.  Moreover, assuming the Unique Games Conjecture (UGC), the $2$-approximation of Veldt et al.~\cite{veldt-18-community-detection} is already optimal.

\begin{theorem}\label{thm:p-time}
  Let $\epsilon>0$ be a fixed constant.
  It is NP-hard to approximate Cluster Deletion within a factor of~$\sqrt{2}-\epsilon$.
  Assuming the UGC, it is NP-hard to approximate Cluster Deletion within a factor of~$2-\epsilon$.
\end{theorem}

Our hardness results are based on a connection between optimal cluster deletion solutions and maximum cliques.
It is well known that an optimal cluster deletion solution does not need to preserve a maximum clique of the input graph, though greedily picking a largest clique of the remaining graph leads to a $2$-approximation~\cite{dessmark-07-cluster-edge-deletion}.\footnote{This does not yield a polynomial-time algorithm, of course, since finding a maximum clique is NP-hard.}
For example, let $G$ be a graph whose vertex set can be partitioned into two cliques, $X = \{x_{1}, \ldots, x_{5}\}$ and $Y = \{y_{1}, \ldots, y_{5}\}$,
and suppose that the only edges between $X$ and $Y$ are
\[
  \{x_i y_j \mid 1\le i,j\le 3\}.
\]
Then
\[
  \{x_1,x_2,x_3,y_1,y_2,y_3\}
\]
is the unique maximum clique of~$G$, of size~$6$.  Nevertheless, the unique optimal cluster deletion solution deletes all nine edges between $X$ and $Y$, thereby splitting this maximum clique.

The key observation behind our reduction is that this behavior changes after adding sufficiently many universal vertices.
If we extend a graph~$G$ by adding a large number~$t$ of universal vertices, then every optimal cluster deletion solution in the resulting graph must place these universal vertices together with a maximum clique of~$G$.
Thus, from a cluster deletion solution of the extended graph, one can recover a maximum clique of~$G$.
This gives a reduction from Clique, or equivalently from Vertex Cover in the complement graph, to Cluster Deletion.

More specifically, in the extended graph, the dominant part of the cluster deletion cost consists of edges between the added universal vertices and the original vertices not placed in the selected clique.
When~$t$ is sufficiently large, the optimal deletion cost is essentially proportional to the number of original vertices outside a maximum clique, i.e., to the vertex cover number of the complement of~$G$.
The reduction is summarized in the following theorem.

\begin{theorem}\label{thm:reduction}
  Let~$c>1$ be a constant.
  If there exists an $f(N)$-time $c$-approximation algorithm for Cluster Deletion on $N$-vertex graphs, then there exists an $ O(f(c'n^2)+n^4)$-time $c$-approximation algorithm for Vertex Cover on $n$-vertex graphs, where~$c'$ is a constant depending only on~$c$.
\end{theorem}

Theorem~\ref{thm:p-time} follows from Theorem~\ref{thm:reduction} together with the known inapproximability results for Vertex Cover~\cite{khot-08-vc-ugc,khot-18}.
Our reduction also has consequences for restricted graph classes.
We say that a graph class~$\mathcal{G}$ is \emph{closed under adding universal vertices} if, for every $G\in\mathcal{G}$, the graph obtained from~$G$ by adding one new vertex adjacent to every vertex of~$G$ also belongs to~$\mathcal{G}$.

\begin{theorem}\label{thm:4}
  Let~$\mathcal{G}$ be a graph class.
  If Clique is NP-hard on~$\mathcal{G}$ and~$\mathcal{G}$ is closed under adding universal vertices, then Cluster Deletion is NP-hard on~$\mathcal{G}$.
\end{theorem}

Combining this theorem with the classical result of Alekseev~\cite{alekseev-82-independent-set} implies NP-hardness of Cluster Deletion on many graph classes.
Since Cluster Deletion is a special case of Constrained Clustering, all of our hardness results also apply to Constrained Clustering.

Finally, we discuss Cluster Deletion and Constrained Clustering on sparse graphs.
Komusiewicz and Uhlmann~\cite{komusiewicz-12-locally-bounded} observed that a $4$-regular graph~$G$ can be partitioned into vertex-disjoint triangles if and only if $G$ has a cluster deletion set of size~$|V(G)|$; see also~\cite{bansal-04-correlation-clustering}.
This implies that Cluster Deletion is NP-hard even on bounded-degree graphs.
However, our reduction always produces dense graphs, and therefore cannot rule out better approximation algorithms for sparse instances.
Indeed, better-than-$2$ approximations are possible for bounded-degree graphs, and more generally for the special case in which all clusters are required to have bounded size~\cite{puleo-15-bounded-cluster-sizes}.

\begin{theorem}\label{thm:better-than-2}
  There exists a polynomial-time $1.92$-approximation algorithm in each of the following cases:
  \begin{itemize}
    \item Cluster Deletion when the clique number of the input graph is bounded by a constant; or
    \item Constrained Clustering when all clusters are required to have bounded size.
  \end{itemize}
\end{theorem}

Most of the best known algorithms for Cluster Editing and Cluster Deletion rely on complex rounding schemes.
Even the simplest LP-based approaches require solving a linear program with $\Theta(n^3)$ constraints~\cite{ailon-08-aggregating-inconsistent-information, zuylen-09-deterministic-approximation}.
This has motivated a line of work on more scalable, and even purely combinatorial, algorithms, at the cost of slightly worse approximation guarantees~\cite{veldt-22-clustering-via-strong-triadic-closure, bengali-23-faster-clustering, makarychev-23-single-pass, cao-24-polylogarithmic-rounds, cao-24-cluster-deletion, veldt-26-constrained-correlation-clustering}.
A common ingredient in these algorithms is a relaxation based on local obstructions: induced $P_3$'s for Cluster Deletion and bad triangles for Cluster Editing.

We use a \emph{bad triangle} to denote the two edges together with the missing edge of an induced $P_3$ (path on three vertices).\footnote{Cluster Editing can be alternatively formulated on signed complete graphs, where missing edges are treated as negative edges. In that formulation, known as \emph{Correlation Clustering}, the instance is a complete graph and a bad triangle is literally a triangle.}
The \emph{Bad Triangle Transversal} problem asks for a minimum-size set~$S$ of edges and missing edges such that every bad triangle intersects~$S$.
By definition, every cluster editing solution is a solution to Bad Triangle Transversal, but the converse does not hold in general.

The aforementioned approximation algorithms use different methods to transform bad triangle transversals into cluster editing solutions.
Their performance therefore depends on the maximum possible gap between the two optimal values.
Veldt~\cite{veldt-22-clustering-via-strong-triadic-closure} showed that the optimal Cluster Editing value is at most twice the optimal bad triangle transversal.
Adriaens and Tatti~\cite{adriaens-26-bad-triangle-transversals} improved this factor to $1.5$, and asked whether the two optimum values are always equal.
We give an explicit small example: a $31$-vertex graph~$G$ with the former strictly larger.

\section{The Reduction from Vertex Cover}

All graphs discussed in this paper are undirected and simple.  The vertex set and edge set of a graph $G$ are denoted by, respectively, $V(G)$ and $E(G)$.
For a subset $U\subseteq V(G)$, denote by $G[U]$ the subgraph of $G$ induced by $U$, and by $G - U$ the subgraph $G[V(G)\setminus U]$, which is further shortened to $G - v$ when $U = \{v\}$.
The \emph{neighborhood} of a vertex $v$ in $G$, denoted by $N_{G}(v)$, comprises vertices adjacent to $v$, i.e., $N_{G}(v) = \{ u \mid u v\in E(G) \}$, and the \emph{closed neighborhood} of $v$ is $N_{G}[v] = N_{G}(v) \cup \{ v \}$.
We omit the subscript when the graph is clear from context.
Two vertices $u$ and $v$ are true twins in $G$ if $N[u] = N[v]$; note that true twins are necessarily adjacent.
A \emph{clique} is a set of pairwise adjacent vertices, and an \emph{independent set} is a set of pairwise nonadjacent vertices.
A graph $G$ is \emph{complete} if $V(G)$ is a clique.
A vertex $v$ is \emph{universal} if $N[v] = V(G)$.  

Our reduction is based on adding sufficiently many universal vertices.  This is known as the join of $G$ and a complete graph in graph-theoretic terminology.

\begin{definition}[Construction]
  Let~$G$ be a graph and let~$t$ be a positive integer.  The graph~$\Gt$ is obtained from~$G$ by adding a clique~$U$ of~$t$ new vertices and making every vertex of~$U$ adjacent to every vertex of~$V(G)$.
\end{definition}

Let~$\OPT(G)$ denote the size of a minimum cluster deletion set of~$G$, and let $\omega(G)$ denote the size of a maximum clique of~$G$.

\begin{proposition}\label{lem:vc-ce}
  Let~$n = |V(G)|$ and~$m = |E(G)|$.  If $m > 0$, then
  \[
    \OPT(\Gt) \le 
      t(n-\omega(G)) + m - \binom{\omega(G)}{2}
      \le
      t(n-\omega(G)) + m.
    \]
    Moreover, if $m>0$, then the second inequality is strict.
\end{proposition}

\begin{proof}
  Let $H = \Gt$.
  Let~$S$ be a maximum clique of~$G$, and let $V_{-} = V(G)\setminus S$.
  Delete all edges of $H$ incident to vertices of $V_-$.  The resulting graph is a cluster graph: one cluster is $U\cup S$, and every vertex of $V_-$ is an isolated singleton cluster.
  The cost of this solution is
  \[
    |U| \cdot |V_{-}| + |E({G})\setminus E({G}[S])| =
    t (n - \omega(G)) + m - \binom{|S|}{2}
    = t(n-\omega(G)) + m - \binom{\omega(G)}{2}.
\]
  It is strictly smaller than $t (n - \omega(G)) + m$ because~$\omega(G) > 1$ when $m > 0$.
\end{proof}

By construction, the vertices of~$U$ are true twins and are universal in~$\Gt$.
It is well known that optimal solutions for cluster editing preserve true twins.  We need the following deletion-only variant.  For completeness, we include a proof.

\begin{proposition}[Folklore]\label{lem:true-twins}
  Let~$E_-$ be a cluster deletion set of a graph~$G$.  There exists a cluster deletion set~$E'_-$ such that~$|E'_-|\le |E_-|$ and every pair of true twins of~$G$ remains adjacent in~$G-E'_-$.
\end{proposition}
\begin{proof}
  If every pair of true twins remains adjacent in~$G-E_-$, then there is nothing to prove.  Otherwise, let~$uv\in E_-$ be an edge whose endpoints are true twins in~$G$.  Let~$C_u$ and~$C_v$ be the clusters containing~$u$ and~$v$, respectively, in~$G-E_-$.  Since $uv$ is deleted, we have $C_u\neq C_v$.

  Because~$u$ and~$v$ are true twins,
  \[
    C_u\cup C_v \subseteq N[u] = N[v].
  \]
  Assume without loss of generality that~$|C_u|\ge |C_v|$.  We modify the
  clustering by moving~$v$ from~$C_v$ to~$C_u$.  In terms of deleted edges, this
  replaces~$E_-$ by
  \[
    E'_-
    =
    E_-
    \cup \{vx \mid x\in C_v\setminus \{v\}\}
    \setminus \{vx \mid x\in C_u\}.
  \]
  The resulting graph is still a cluster graph: $v$ is now joined to every vertex of~$C_u$, and the remaining vertices of~$C_v\setminus\{v\}$ still form a clique.  Moreover, since only adjacencies between $v$ and~$C_u\cup C_v$ are changed,
  \[
    |E'_-|
    =
    |E_-| + |C_v|-1-|C_u|
    \le
    |E_-|-1
    <
    |E_-|.
  \]

  We repeat this operation as long as some pair of true twins is separated. Each iteration strictly decreases the size of the deletion set, so the process terminates.  The final deletion set has size at most~$|E_-|$ and preserves all pairs of true twins.
\end{proof}

Some remarks are in order.  In the weighted setting, one may simply merge a set of true twins into a single ``supernode.''  Similar observations on true twins have been used in parameterized algorithms~\cite{cao-12-kernel-cluster-editing,cao-22-edge-modification}.  They are less visible in approximation algorithms because, in the standard LP formulations, the edge between any pair of true twins receives value~$0$ in an optimal fractional solution; hence such vertices are never separated by the rounding.\footnote{Indeed, $x_{u v} = 0$ if and only if $u$ and~$v$ are true twins in the revised graph with edge set $E(G) \Delta \{e\mid x_{e} = 1\}$, with $\Delta$ denoting symmetric difference.}

\begin{lemma}\label{lem:ce-vc}
  Let~$n=|V(G)|$.  Given any cluster deletion set~$E_-$ of~$\Gt$, one can produce, in time $O((t+n)^2)$, a clique of~$G$ of size at least
  \[
    n-\frac{|E_-|}{t}.
  \]
\end{lemma}
\begin{proof}
  Let~$H=\Gt$, and let~$U$ be the set of the $t$ added universal vertices.  By Proposition~\ref{lem:true-twins}, we may transform~$E_-$ into a cluster deletion set~$E'_-$ such that~$|E'_-|\le |E_-|$ and all vertices of~$U$ remain pairwise adjacent in~$H-E'_-$.  Since~$H-E'_-$ is a cluster graph, this means that all vertices of~$U$ belong to a single cluster; denote this cluster by~$C$.

  We return the clique
  \[
    K = C\setminus U.
  \]
  Because $C$ is a clique in $H-E'_-$ and no edge inside $V(G)$ is introduced, $K$ is a clique of $G$.
  Every vertex in~$V(G)\setminus K$ lies outside the cluster containing~$U$.
  Therefore, all edges between~$U$ and~$V(G)\setminus K$ must be deleted by
  $E'_-$.  Hence
  \[
    |E'_-|
    \ge
    |U|\cdot |V(G)\setminus K|
    =
    t(n-|K|).
  \]
  Since~$|E'_-|\le |E_-|$, we obtain
  \[
    |K|
    \ge
    n-\frac{|E'_-|}{t}
    \ge
    n-\frac{|E_-|}{t}.
  \]

  Given $E_-$ explicitly, the clustering of $H-E_-$ and the set $K$ can be computed in $O(|V(H)|^2)$ time.
\end{proof}

Proposition~\ref{lem:vc-ce} and Lemma~\ref{lem:ce-vc} already imply the NP-hardness of Cluster Deletion: it suffices to set~$t= |E(G)|$.

\begin{corollary}\label{cor:reduction}
  Let~$m = |E(G)|$.
  From any optimal cluster deletion set~$E_-$ of~$\Gm$, one can produce a maximum clique of~$G$ in polynomial time.
\end{corollary}
\begin{proof}
  Let~$n=|V(G)|$ and~$m=|E(G)|$.  If~$m=0$, then any single vertex is a maximum clique.  Hence assume that~$m>0$.
  Applying Lemma~\ref{lem:ce-vc} with~$t=m$, we obtain a clique~$K$ of~$G$ such that~$|K| \ge n - \frac{|E_{-}|}{m}$.
  Since~$E_-$ is optimal, Proposition~\ref{lem:vc-ce} gives
  \[
    |K| > n - \frac{m  (n - \omega(G)) + m}{m}
    = \omega(G) - 1.
  \]
  Since~$|K|$ is an integer, we conclude that $|K|=\omega(G)$.  Thus~$K$ is a maximum clique.
\end{proof}

\begin{proof}[Proof of Theorem~\ref{thm:4}]
  Let $G\in\mathcal{G}$ be an instance of Clique, and let $m=|E(G)|$.
  Since $\mathcal{G}$ is closed under adding one universal vertex, repeated application of the closure property gives $G\vee K_m\in\mathcal{G}$.
  If Cluster Deletion could be solved in polynomial time on $\mathcal{G}$, then we could compute an optimal cluster deletion set of $G\vee K_m$ and, by Corollary~\ref{cor:reduction}, recover a maximum clique of $G$ in polynomial time. This contradicts the NP-hardness of Clique on $\mathcal{G}$.
\end{proof}

For the main theorem, we use the following standard relation between cliques and vertex covers.  The \emph{complement graph}~$\overline{G}$ of a graph~$G$ is defined on the same vertex set~$V(G)$, where two distinct vertices~$u$ and~$v$ are adjacent in~$\overline{G}$ if and only if~$uv\notin E(G)$.
Note that the complement of~$\overline{G}$ is~$G$.
Let~$\tau(G)$ denote the size of a minimum vertex cover of~$G$.  Since a clique in a graph is an independent set in its complement, we have
\[
  \tau(G) = |V(G)| - \omega \left(\overline{G} \right).
\]

\begin{proof}[Proof of Theorem~\ref{thm:reduction}]
  Let~$G$ be the input graph for Vertex Cover, and let~$n=|V(G)|$.  If~$G$ is edgeless, then the empty set is an optimal vertex cover.  If~$G$ is complete, then any set of~$n-1$ vertices is an optimal vertex cover.  Thus we may assume that~$G$ has at least one edge and is not complete.

  Choose an arbitrary edge~$v_1v_2\in E(G)$.  For~$i=1,2$, define
  \[
    G_{i} = {G - v_{i}} \quad \text{ and }\quad
    H_{i} = \overline{G_{i}} \vee K_{t},
  \]
  where 
  \[
t = \left\lceil \frac{c n^{2}}{2(c-1)} \right\rceil.
  \]
  Since every vertex cover of~$G$ contains at least one endpoint of the edge~$v_1v_2$, there is an index~$i\in\{1,2\}$ such that deleting~$v_i$ from a minimum vertex cover of~$G$ leaves a vertex cover of~$G_i$.  Hence
  \[
    \min_{i \in \{1, 2\}}  \tau(G_{i}) \le \tau(G) - 1.
  \]

  For each~$i\in\{1,2\}$, run the assumed $c$-approximation algorithm for Cluster Deletion on~$H_i$, and let~$E_i$ be the returned deletion set.  By the approximation guarantee and Proposition~\ref{lem:vc-ce},
  \[
    |E_{i}| \le c\cdot \OPT(H_{i}) \le c \left( t \left( |V(G_i)| -\omega \left(\overline{G_{i}} \right) \right) + \left| E(\overline{G_{i}})\right| \right) < c t \tau(G_{i}) + c n^{2} / 2.
  \]
  We may assume without loss of generality that $|E_{1}| \le |E_{2}|$.  Then
  \[
|E_{1}|
    = \min_{i \in \{1, 2\}} |E_{i}|
    < c t \cdot \min_{i \in \{1, 2\}} \tau(G_{i}) + \frac{c n^2}{2} \le c t \left( \tau(G) - 1 \right) + \frac{c n^2}{2}.
  \]
  Since~$c n^{2} / 2 \le t(c-1)$ by the choice of~$t$, we have
  \[
    |E_{1}| < c t (\tau(G) - 1) + t(c-1) = c t\, \tau(G) - t.
  \]

  Now apply Lemma~\ref{lem:ce-vc} to~$E_1$ in~$H_1$.  This yields a clique $K$ of~$\overline{G_1}$ such that
  \[
    |K| \ge \left| V \left( \overline{G_{1}} \right) \right| - \frac{|E_{1}|}{t} \ge n - 1 - \frac{c t\, \tau(G) - t}{t} = n - c\tau(G).
  \]
  Since~$K$ is a clique in~$\overline{G_1}$, it is an independent set in~$G_1$ and~$G$.  Hence, $V(G)\setminus K$ is a vertex cover of~$G$.
  The size of this vertex cover is
  $
    |V(G)\setminus K|
    =
    n-|K|
    \le c\tau(G).
  $
  Thus, we obtain a $c$-approximation for Vertex Cover.

  It remains to analyze the running time.  Let~$c' = \left\lceil \frac{c}{2(c-1)} \right\rceil + 1.$  Then
  each graph~$H_i$ has
  \[
    |V(H_i)| = n-1+ t \le c' n^2
  \]
  vertices.  Constructing~$\overline{G_i}$ takes $O(n^2)$ time, and explicitly constructing the join with~$K_t$ takes
  $
    O(|V(H_i)|^2)=O(n^4)
  $
  time.
  The assumed approximation algorithm is invoked twice, each time on a graph with at most~$c'n^2$ vertices.  Thus, the total time spent in the approximation algorithm is $2 f(c'n^2)$.  Finally, applying Lemma~\ref{lem:ce-vc} takes
  $ O((t+n)^2)=O(n^4) $ time.  Therefore, the total running time is
  \[
    O\left( f(c' n^{2}) + n^4 \right),
  \]
  as claimed.
\end{proof}

Theorem~\ref{thm:p-time} follows immediately from Theorem~\ref{thm:reduction}.  The NP-hardness within $\sqrt2-\epsilon$ follows from the inapproximability of Vertex Cover due to Khot et al.~\cite{khot-18}; the UGC-hardness within $2-\epsilon$ follows from Khot and Regev~\cite{khot-08-vc-ugc}.

\section{Further consequences and discussion}

We now discuss some implications of Theorem~\ref{thm:p-time}.

\subsection{Edge modification problems}

In an edge modification problem, we are asked to modify at most $k$ edges of a given graph $G$ to make the graph satisfy a certain property.  The most common of the operations are edge deletions, additions (also known as completion), and their combinations.
Since cluster graphs are precisely $P_3$-free graphs, edge modification problems to cluster graphs are the simplest of all nontrivial edge modification problems.  Note that edge modification problems toward $P_2$-free graphs, i.e., edgeless graphs, are trivial.

Our results suggest that, from the viewpoint of approximation, Cluster Deletion may be harder than Cluster Editing, despite the widely held belief that the former is easier.
From the perspective of algorithm designers, a single type of operation is ostensibly easier to handle than multiple kinds.
However, similar phenomena have long been observed in edge modification problems.
A graph is a \emph{split graph} if its vertex set can be partitioned into a clique and an independent set.
A classical result of Hammer and Simeone~\cite{hammer-81-splittance} gave a linear-time algorithm for the Split Editing problem.
In contrast, Split Deletion is NP-hard~\cite{natanzon-01-edge-modification}.
We also note that the APX-hardness of Cluster Deletion~\cite{shamir-04-graph-modification-problem} was shown before that of Cluster Editing~\cite{charikar-04-approximate-clustering}, with a far simpler reduction.

We remark that the Cluster Completion problem is trivial, because one has to add all the missing edges in each component, and the Split Completion problem is equivalent to the Split Deletion problem, because the complement of a split graph is also a split graph.
It is worth exploring whether these two classes are exceptions, or whether there is a more general explanation for why editing can be easier than one-sided modification variants.

\subsection{Relation to Multicut}

In the Multicut problem, we are given a supply graph $G$ and a demand graph $H$ with $V(H)\subseteq V(G)$, and asked for a minimum set of edges $E_{-}\subseteq E(G)$ so that no endpoints of an edge in $H$ are connected in $G - E_{-}$.
Cluster Deletion can be viewed as Multicut with supply graph $G$ and demand graph $\overline{G}$.

Karzanov~\cite{karzanov-89-polyhedra-multicommodity} observed that~$H$ decides the fractionality of the polytope of the standard LP formulation:
\begin{equation}
  \label{eq:1}
\left\{
    x\in \mathbb{R}_{\ge 0}^{E(G)} :
    \sum_{e\in P} x_e \ge 1 \text{ for all $s$--$t$ path $P$ with } st\in E(H)
  \right\}.
\end{equation}
If $H$ has a certain property, then the polytope~\eqref{eq:1} is $\frac{1}{4}$-integral (i.e., every basic solution is a multiplier of $\frac{1}{4}$); otherwise, for every positive integer $k$, there exists a graph~$G$ such that the polytope is not $\frac{1}{k}$-integral.

Naturally, the demand graph $H$ also affects the approximability of the Multicut problem.
Near-tight inapproximability results are known when $H$ is a matching and when $H$ is a complete graph.
Our result settles the case in which the demand graph is $\overline{G}$, the complement of the supply graph.
An interesting question is whether these approximation phenomena follow a pattern analogous to Karzanov's characterization of the fractionality of the multicut polytope~\cite{karzanov-89-polyhedra-multicommodity}.

\subsection{Better-than-$2$ approximations in restricted settings}

The cluster LP~\cite{cao-24-lp, cao-25-sublinear}, reproduced as~\eqref{fig:constrained_cluster_lp}, considers the assignment of vertices to potential clusters.
It has a variable~$x_{uv}$ for every vertex pair $u, v \in V(G)$ and a variable~$z_S$ for every nonempty vertex set~$S \subseteq V(G)$.
The variable $x_{uv}$ is 1 if $uv$ is not an edge of the resulting graph, and the variable $z_S$ is 1 if $S$ is a cluster in the final clustering.

\begin{equation}
  \label{fig:constrained_cluster_lp}
  \begin{array}{r@{\quad}l@{\qquad}l}
    \min & \displaystyle\sum_{uv \in E(G)} x_{uv} + \displaystyle\sum_{uv \not\in E(G)} (1 - x_{uv}) & \\[1.5em]
    \text{s.t.} & \displaystyle\sum_{\mathclap{S \ni u}} z_S = 1 & \forall u \in V, \\[1.5em]
         & \displaystyle\sum_{\mathclap{S \supseteq \{u,v\}}} z_S = 1 - x_{uv} & \forall u, v \in V(G), \\[1.5em]
         & z_S \ge 0 & \forall S \subseteq V, S \neq \emptyset. 
\end{array}
\end{equation}

Since $S$ ranges over all nonempty subsets of $V(G)$, this LP has exponentially many variables.  The first obstacle is thus to solve it (approximately) in polynomial time.  Cao et al.~\cite{cao-24-lp} gave a polynomial-time approximation scheme (PTAS), followed by a very nontrivial rounding step.
It is easy to extend LP~\eqref{fig:constrained_cluster_lp} to Constrained Clustering by adding constraints to enforce the hard constraints of the instance on the set $F$ of friendly pairs and the set $H$ of hostile pairs:
\begin{align*}
  x_{uv} = 0 &\quad \forall uv \in F,
  \\
  x_{uv} = 1 &\quad \forall uv \in H.
\end{align*}
Kalavas et al.~\cite{kalavas-25-constrained-correlation-clustering} developed a rounding algorithm that turns an (approximate) solution of the extended LP into a solution for the Constrained Clustering instance with a loss of at most~$1.92$.
Thus, a PTAS for solving the extended LP would imply a~$(1.92+\epsilon)$-approximation for Constrained Clustering.
Consequently, assuming the UGC, there cannot be a polynomial-time scheme that computes a sufficiently accurate approximate solution to the extended LP in full generality; otherwise, combined with the rounding of Kalavas et al., it would yield a better-than-$2$ approximation for Constrained Clustering, violating Theorem~\ref{thm:p-time}.

Kalavas et al.~\cite{kalavas-25-constrained-correlation-clustering} tried to build a randomized PTAS for the new LP but did not succeed.\footnote{To exclude randomized PTAS using Theorem~\ref{thm:reduction}, one needs stronger complexity assumptions, which we do not elaborate.}  Shi Li has also indicated that such a PTAS is unlikely even for highly restricted choices of $F$ and $H$.\footnote{Private communication.}
However, their rounding scheme does imply a better-than-2 approximation when the extended LP can be solved in polynomial time, in particular, when it has polynomial size.

\begin{proof}[Proof of Theorem~\ref{thm:better-than-2}]
  We use the extended cluster LP together with the rounding algorithm of
  Kalavas et al.~\cite{kalavas-25-constrained-correlation-clustering}.
  Their rounding algorithm converts a fractional solution of the extended LP into an integral solution whose cost is at most $1.92$ times the LP value.
  Therefore, it suffices to show that, in the two cases considered here, the extended LP has polynomial size and can be solved in polynomial time.

  The cluster LP has a variable $z_S$ for every nonempty set $S\subseteq V(G)$,
  indicating whether $S$ is chosen as a cluster. If every feasible cluster has
  size at most $k$, where $k$ is a fixed constant, then all variables $z_S$ with
  $|S|>k$ can be omitted. The number of remaining variables is
  \[
    \sum_{i=1}^k \binom{n}{i}=O(n^k),
  \]
  which is polynomial for fixed $k$. The resulting LP also has polynomially many
  constraints, and hence can be solved in polynomial time. Applying the rounding
  algorithm of Kalavas et al. gives a polynomial-time $1.92$-approximation.
  This proves the bounded-cluster-size case.

  Now consider Cluster Deletion on an input graph $G$ with clique number at most $k$. Since edge additions are not allowed, every cluster in any feasible solution must already be a clique of $G$.
Therefore every feasible cluster has size at most $\omega(G)\le k$. Thus the same polynomial-size LP argument applies, and the rounding algorithm gives a polynomial-time $1.92$-approximation.
\end{proof}

\section{Cluster Editing vs.~Bad Triangle Transversal}

For a graph~$G$, let~$\ce(G)$ and~$\btt(G)$ denote the minimum solution size of Cluster Editing and Bad Triangle Transversal, respectively.
We define the Cluster Editing--Bad Triangle Transversal ratio as
\[
  \rho = \sup_{G} \frac{\ce(G)}{\btt(G)}.
\]
Adriaens and Tatti~\cite[Theorem 1.4]{adriaens-26-bad-triangle-transversals} showed that
\[
  1\le \rho \le 1.5,
\]
and asked whether $\rho = 1$.

We give an example showing that
\[
   \rho \ge \frac{10}{9}.
\]
For ease of presentation, we first describe a weighted graph in Figure~\ref{fig:cc-ne-btt}, and then explain how to obtain an equivalent unweighted instance.
Each vertex $v$ has an integer weight $w(v)$, and the cost of modifying a pair $uv$ is $w(u)w(v)$.
For example, modifying the pair $x_1 y_1$ has cost $20$.

\begin{figure}[ht!]
  \centering
    \begin{tikzpicture}
      \node[empty vertex, "$x_0$" below] (x0) at (0, 0) {$4$};
      \foreach \i in {1, 2, 3} {
        \pgfmathsetmacro{\ang}{120*\i - 30}
        \node[empty vertex] (x\i) at (\ang:2) {$4$};
        \node at (\ang:2.5) {$x_{\i}$};
        \pgfmathsetmacro{\ang}{120*\i - 10}
        \node[empty vertex] (y\i) at (\ang:1.) {$5$};
        \node at ({\ang+25}:.85) {$y_{\i}$};
        \draw (x0) -- (y\i) -- (x\i) -- (x0);
      }
      \draw (x1) -- (x2) -- (x3) -- (x1);
    \end{tikzpicture}
\caption{A weighted graph.}
  \label{fig:cc-ne-btt}
\end{figure}
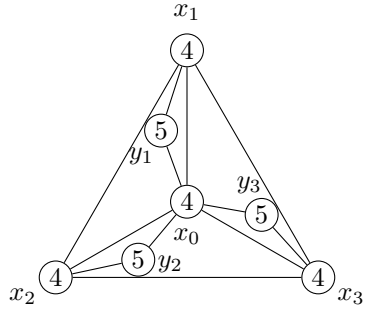

There are 15 induced $P_3$'s in Figure~\ref{fig:cc-ne-btt}:
\[
  x_{i} x_{0} y_{j},
  \quad
  x_{i} x_{j} y_{j},
  \quad
  y_{i} x_{0} y_{j},
  \quad  i, j\in \{1, 2, 3\}, \; i \ne j.
\]
Thus, an optimal solution to Bad Triangle Transversal is
\[
  \{x_{0} y_{i}\mid 1\le i \le 3\} \cup 
  \{x_{i} x_{j}\mid 1\le i < j \le 3\},
\]
with value
\[
  4 \times 5 \times 3 + 4 \times 4 \times 3 = 108.
\]

There are several optimal clusterings, one of which is
\[
  \{x_{0}, x_{1}, x_{2}, x_{3}\},
  \{y_{1}, y_{2}, y_{3}\}.
\]
The modified edges are all the edges between them,
with value
\[
  4 \times 5 \times 6 = 120.
\]

To obtain an unweighted graph as claimed, we replace each vertex of weight~$w$ by a set of~$w$ true twins.
We leave it to the reader to verify that Proposition~\ref{lem:true-twins} also holds for Cluster Editing, and that an analogous statement holds for Bad Triangle Transversal.

The following observation connects the approximation thresholds for these two problems as well as $\rho$.

\begin{proposition}\label{lem:ce-btt-reduction}
  Let~$c>1$ be a constant.  If there exists a polynomial-time $c$-approximation algorithm for Cluster Editing, then there exists a polynomial-time $c\rho$-approximation algorithm for Bad Triangle Transversal.
\end{proposition}
\begin{proof}
  Given a graph~$G$, we use the assumed algorithm to compute an approximate cluster editing set~$S$.  Then
  \[
    |S| \le c\cdot \ce(G)
    \le c\rho\cdot \btt(G).
  \]
Since $S$ is also a solution to Bad Triangle Transversal, this gives a $c\rho$-approximation for Bad Triangle Transversal.
\end{proof}

On the one hand, a better-than-$\frac{2}{\rho}$ approximation for Cluster Editing would imply a better-than-2 approximation for Bad Triangle Transversal.
On the other hand, if Bad Triangle Transversal is hard to approximate within a factor $c$, then Cluster Editing is hard to approximate within a factor $c/\rho$.
For example, if Bad Triangle Transversal cannot be approximated within $2$, under certain assumptions, Cluster Editing cannot be approximated within $4/3$ under certain assumptions.  Here we are using the fact~$\rho \le 1.5$, and hence the lower bound can be further improved if we have a better bound for~$\rho$.

\paragraph{Acknowledgment.}
We are grateful to Shi Li for pointing out a bug in an earlier version of the paper and for answering our questions about~\cite{cao-24-lp,cao-25-sublinear}, and to Nate Veldt for helpful discussions, in particular for bringing \cite{kalavas-25-constrained-correlation-clustering} to our attention.

\bibliographystyle{plainurl}

\end{document}